\documentclass[10pt,conference,letterpaper]{IEEEtran}
\usepackage{epsfig}
\usepackage{graphicx}
\usepackage{cite}
\usepackage{amsmath}
\usepackage{amssymb}
\usepackage{verbatim}
\usepackage[top=0.75in, bottom=0.7in, left=0.75in, right=0.75in]{geometry}
\def\BE{\begin{equation}}
\def\EE{\end{equation}}
\def\BEA{\begin{eqnarray}}
\def\EEA{\end{eqnarray}}

\newtheorem{thm}{Theorem}
\newtheorem{lem}[thm]{Lemma}

\newtheorem{corol}[thm]{Corollary}

\newcommand{\diag}{\mathop{\bf diag}}

\newcommand{\half}{ {-\tfrac{1}{2}} }
\newcommand{\lint}[1]{ {\int \limits_{#1}}}

\newcommand{\lexp}[1]{ {\exp\{\half #1\}}}

\newcommand{\ignore}[1]{{}}

\newcommand\etal{{\textsl{et al.\,\,}}}

\newcommand\vl{{\bf l}}

\newcommand\ns{{\normalsize}}
\newcommand\0{{\mathbf{0}}}

\newcommand\R{\mathbb{R}}
\newcommand\Z{\mathbb{Z}}
\newcommand\N{\mathcal{N}}
\newcommand\nbr{\Gamma}
\usepackage{color}
\usepackage{varioref}
\begin{document}

\twocolumn

\title{A Low Density Lattice Decoder\\ via Non-parametric Belief Propagation}

\author{
\authorblockN{Danny Bickson}
\authorblockA{IBM Haifa Research Lab\\
Mount Carmel, Haifa 31905, Israel\\
Email: danny.bickson@gmail.com}
%\and
%\authorblockN{J. K. Johnson and M. Chertkov}
%\authorblockA{Complex Systems Group,\\
%T-13, Theoretical Division,\\
%Los Alamos National Laboratory\\
%Los Alamos, NM 87545\\
%Email: \{jasonj,chertkov\}@lanl.gov}
\and
\authorblockN{Alexander T. Ihler}
\authorblockA{Bren School of Information \\and Computer Science\\
University of California, Irvine\\
Email: ihler@ics.uci.edu}
\and
\authorblockN{Harel Avissar and Danny Dolev}
\authorblockA{School of Computer Science and Engineering\\
Hebrew University of Jerusalem\\
Jerusalem 91904, Israel\\
Email: \textbraceleft harela01,dolev\textbraceright@cs.huji.ac.il}
}

%\IEEEpeerreviewmaketitle
\maketitle

\begin{abstract}
The recent work of Sommer, Feder and Shalvi presented a new family of codes called low density lattice codes (LDLC)
that can be decoded efficiently and approach the capacity of the AWGN channel. A linear time iterative decoding
scheme which is based on a message-passing formulation on a factor graph is given.

  In the current work we report our theoretical findings regarding the relation between the LDLC decoder and belief propagation. We show that the LDLC decoder is an instance of non-parametric belief propagation and further connect it to the Gaussian belief propagation algorithm. Our new results enable borrowing knowledge from the non-parametric and Gaussian belief propagation domains into the LDLC domain. Specifically, we give more general convergence conditions for convergence of the LDLC decoder (under the same assumptions of the original LDLC convergence analysis). We discuss how to extend the LDLC decoder from Latin square to full rank, non-square matrices. We propose an efficient construction of sparse generator matrix and its matching decoder. We report preliminary experimental results which show our decoder has comparable symbol to error rate compared to the original LDLC decoder.%

%We provide experimental results for the multiuser detection problem
%showing that the LDLC decoder has superior performance relative to a recent iterative decoder.
%\footnotetext[1]{Contributed equally to this work.}
\end{abstract}

%%%%%%%%%%%%%%
% alex suggestion for outline
%%%%%%%%%%%%%%
% Here's an alternative outline I would suggest:
% (1) Intro
% (2) Lattice codes
% (3) Factor graphs and sum-product, general definitions
%    (a) Gaussian formulae, as in table 1.  It would be useful to show how to map a factor graph representation, e.g. factors like N(x_i; y_i, \sigma^2)  and  \delta( h'x = b ), into the Gaussian messages.
%    (b) Gaussian mixture formulae, a la NBP.   Show message form for product of messages.
% 
% 
% Figs 3&4 are attractive but very anecdotal.  They don't seem to be trying to convince the reader of anything.  We need to consider -- what points are we trying to make with the experiments section?  Computational efficiency, performance for decoding lattice codes, ?  These are the things the experiments should be trying to highlight.
% 
% About the "sufficient conditions for convergence" -- I say again that this is a very slippery claim.  The conditions are sufficient for the convergence of GaBP, but *not* for Gaussian mixtures (NBP).  It might be *necessary* for GaBP to converge to claim that NBP converges, but unfortunately "a sufficient condition for a necessary condition" is neither necessary nor sufficient.  I do think it's an interesting point, but you have to be careful what you say.
% 
% If you'd like to try to figure out how to re-work this on the phone, let me know & we'll set up a time.
% 
% --Alex
\IEEEpeerreviewmaketitle
\section{Introduction}
Lattice codes provide a continuous-alphabet encoding procedure, in which integer-valued information bits
are converted to positions in Euclidean space.  Motivated by the success of low-density parity check (LDPC)
codes~\cite{BibDB:Gallager}, recent work by %Sommer, Feder and Shalvi \cite{LDLC_Sommer} presented
Sommer \emph{et al.}~\cite{LDLC_Sommer} presented
low density lattice codes (LDLC).  Like LDPC codes, a LDLC code has a sparse decoding matrix
which can be decoded efficiently using an iterative message-passing algorithm defined
over a factor graph.
In the original paper, the lattice codes were limited to Latin squares,
and some theoretical results were proven for this special case.

The non-parametric belief propagation (NBP)\ algorithm is an efficient method for approximated inference on continuous graphical models. The NBP\ algorithm was originally introduced in~\cite{NBP}, but has recently been rediscovered independently in several domains, among them compressive sensing \cite{CSBP06,CSBP2009} and low density lattice decoding \cite{LDLC_Sommer}, demonstrating very good empirical performance in these systems. 

In this work, we  investigate the theoretical relations between the LDLC decoder and belief propagation, and show it is an instance of the NBP algorithm.
This understanding has both theoretical and practical consequences. From the theory point of view we provide a cleaner and more standard derivation of the LDLC update rules, from the graphical models perspective. From the practical side we propose to use the considerable body of research that exists in the NBP\ domain to allow construction of efficient decoders. 

We further propose a new family of LDLC codes as well as a new LDLC decoder based on the NBP algorithm . By utilizing sparse generator matrices rather than the sparse parity check matrices used in the original LDLC work,
we can obtain a more efficient encoder and decoder.
We introduce the theoretical foundations which are the basis of our new decoder and give preliminary experimental results which show our decoder has comparable performance to the LDLC decoder.

%We provide an efficient implementation, utilizing efficient
%algorithm for approximating Gaussian mixture products, that where investigated in the NBP context.

The structure of this paper is as follows. Section \ref{s-ldlc} overviews LDLC codes, belief propagation on factor graph and the LDLC decoder algorithm. Section III rederive the original LDLC algorithm using standard graphical models terminology, and shows it is an instance of the NBP algorithm. Section \ref{s-novel} presents a new family of LDLC codes as well as our novel decoder.  We further discuss the relation to the GaBP algorithm. In Section \ref{s-new-const} we discuss convergence and give more general sufficient conditions for convergence, under the same assumptions used in the original LDLC work. Section VI brings preliminary experimental results of evaluating our NBP decoder vs. the LDLC\ decoder. We conclude in Section VII.

\section{Background}
\label{s-ldlc}
\subsection{Lattices and low-density lattice codes} \label{lattice_codes}

An $n$-dimensional lattice $\Lambda$ is defined by a generator
matrix $G$ of size $n \times n$. The lattice consists of the discrete set of points
$x = (x_1, x_2, . . . , x_n) \in R^n$ with $x = Gb$, where $b \in Z^n$ is the set of all possible integer
vectors.

%!!! shaping region; in keeping with [x] we ignore the shaping region.

A low-density lattice code (LDLC) is a lattice with a non-singular
generator matrix $G$, for which $H = G^{-1}$ is sparse. It is
convenient to assume that $det(H) = 1/ det(G) = 1$. An $(n, d)$
regular LDLC code has an $H$ matrix with constant row and
column degree $d$. % weight?
In a latin square LDLC, the values of the
$d$ non-zero coefficients in each row and each column are some permutation
of the values $h_1, h_2, \cdots , h_d$.

%\subsection{Linear channel model} \label{iter_dec_sec}
We assume a linear channel with additive white Gaussian noise (AWGN).  For a vector of integer-valued
information $b$ the transmitted codeword is $x=Gb$, where $G$ is the LDLC encoding matrix,
and the received observation is $y=x+w$ where $w$ is a vector of i.i.d. AWGN with diagonal
covariance $\sigma^2 I$.  The decoding problem is then to estimate $b$ given the observation vector $y$;
for the AWGN channel, the MMSE estimator is
\small
\BE
 b^* =\arg \mathop{\min }\limits_{b \in \Z^n}||y - Gb||^2\,. \label{eq:ls}
\EE
\normalsize
% Rather than searching for the best value of $b$, an LDLC code is defined on the transmitted codeword $x$
% with $b$ defined implicitly via $b=Hx$.
% Constructing an iterative decoder then proceeds by creating a graph which corresponds to this code,
% finding an estimate $\hat x$ of the codeword $x$ (whose elements take on values in $\R$), and
% rounding $H\hat x$ to estimate $b$.

\subsection{Factor graphs and belief propagation}
Factor graphs provide a convenient mechanism for representing structure among random variables.
Suppose a function or distribution $p(x)$ defined on a large set of variables $x=[x_1,\ldots,x_n]$
factors into a collection of smaller functions $p(x)=\prod_s f_s(x_s)$, where each $x_s$ is a vector 
composed of a smaller subset of the $x_i$.  We represent this factorization as a bipartite graph with ``factor
nodes'' $f_s$ and ``variable nodes'' $x_i$, where the neighbors $\Gamma_s$ of $f_s$ are the variables
in $x_s$, and the neighbors of $x_i$ are the factor nodes which have $x_i$ as an argument ($f_s$ such
that $x_i$ in $x_s$).  For compactness, we use subscripts $s,t$ to indicate factor nodes and $i,j$ to 
indicate variable nodes, and will use $x$ and $x_s$ to indicate sets of variables, typically formed
into a vector whose entries are the variables $x_i$ which are in the set.

The belief propagation (BP) or sum-product algorithm~\cite{BibDB:FactorGraph} is a popular technique for estimating 
the marginal probabilities of each of the variables $x_i$.  BP follows a message-passing formulation, in which
at each iteration $\tau$, every variable passes a message (denoted $M_{is}^\tau$) to its neighboring factors, and factors to their neighboring 
variables.  These messages are given by the general form,
\small
\begin{align} %\label{eq:bpmsg}
 M^{\tau+1}_{is}(x_i) &= f_i(x_i) \prod_{t\in \nbr_i \setminus s} M^\tau_{ti}(x_i)\,, \nonumber \\
 M^{\tau+1}_{si}(x_i) &= \int \limits_{x_s \setminus x_i}  f_s(x_s) \prod \limits_{j \in \nbr_s \setminus i} M^\tau_{js}(x_j)dx_{s}\,. \label{eq:BP}
\end{align}
\normalsize
Here we have included a ``local factor'' $f_i(x_i)$ for each variable, to better parallel our development in the sequel.
When the variables $x_i$ take on only a finite number of values, the messages may be represented as vectors; the
resulting algorithm has proven effective in many coding applications including low-density parity check (LDPC) 
codes~\cite{BibDB:McElieceMacKayCheng}.  In keeping with our focus on continuous-alphabet codes, however, we
will focus on implementations for continuous-valued random variables.

\subsubsection{Gaussian Belief Propagation}
When the joint distribution $p(x)$ is Gaussian, $p(x)\propto \exp\{-\frac{1}{2} x^T J x + h^T x\}$, the
BP messages may also be compactly represented in the same form.  Here we use the ``information form'' of the
Gaussian distribution, ${\cal N}(x;\mu,\Sigma)={\cal N}^{-1}(h,J)$ where $J=\Sigma^{-1}$ and $h=J\mu$.
In this case, the distribution's factors can always be written in a pairwise form, so that each function involves at most two
variables $x_i$, $x_j$, with $f_{ij}(x_i,x_j)=\exp\{- J_{ij} x_i x_j\}$, $j\ne i$, and 
$f_i(x_i)= \exp\{-\frac{1}{2} J_{ii} x_i^2 + h_i x_i\}$.

Gaussian BP (GaBP) then has messages that are also conveniently represented as information-form Gaussian 
distributions.  If $s$ refers to factor $f_{ij}$, we have
\small
\begin{align}
 M^{\tau+1}_{is}(x_i) &= {\cal N}^{-1}(\beta_{i\setminus j},\alpha_{i\setminus j})\,, \nonumber \\
  \alpha_{i \backslash j} &= J_{ii} + \sum_{{k} \in \nbr_i \setminus j} \alpha_{ki}\,, &
  \beta_{i \backslash j} &= h_i + \sum_{k \in \nbr_i \setminus j} \beta_{ki}\,, \label{eq:GaBP1}\\
 M^{\tau+1}_{sj}(x_j) &= {\cal N}^{-1}(\beta_{ij},\alpha_{ij})\,,\nonumber \\
  \alpha_{ij} &= -J_{ij}^2 \alpha_{i \backslash j}^{-1}\,, &
  \beta_{ij} &= -J_{ij} \alpha_{i \backslash j}^{-1} \beta_{i \backslash j}\,. \label{eq:GaBP2}
\end{align} 
\normalsize
From the $\alpha$ and $\beta$ values we can compute the estimated marginal distributions,
which are Gaussian with mean $\hat{\mu}_{i}= \hat{K}_i (h_i + \sum_{k \in \nbr_i} \beta_{ki})$
and variance $\hat{K}_{i}=(J_{ii} + \sum_{{k} \in \nbr_i} \alpha_{ki})^{-1}$.
It is known that if GaBP converges, it results in the exact
MAP estimate $x^*$, although the variance estimates $\hat{K}_i$ 
computed by GaBP are only approximations to the correct
variances \cite{BibDB:Weiss01Correctness}.

\subsubsection{Nonparametric belief propagation}
In more general continuous-valued systems, the messages do not have a simple closed form and must be
approximated.  Nonparametric belief propagation, or NBP, extends the popular class of particle filtering
algorithms, which assume variables are related by a Markov chain, to general graphs.  In NBP, messages are
represented by collections of weighted samples, smoothed by a Gaussian shape--in other words, Gaussian
mixtures.

NBP follows the same message update structure of~\eqref{eq:BP}.  Notably, when the factors are all
either Gaussian or mixtures of Gaussians, the messages will remain mixtures of Gaussians as well,
since the product or marginalization of any mixture of Gaussians is also a mixture of Gaussians~\cite{NBP}.
However, the product of $d$ Gaussian mixtures, each with $N$ components, produces a mixture of $N^d$
components; thus every message product creates an exponential increase in the size of the mixture.
For this reason, one must approximate the mixture in some way.  NBP typically relies on a stochastic
sampling process to preserve only high-likelihood components, and a number of sampling algorithms have 
been designed to ensure that this process is as efficient as possible~\cite{NBP2,briers05,rudoy07}.
One may also apply various deterministic algorithms to reduce the number of Gaussian mixture components~\cite{KDE};
for example, in~\cite{Ihler-JSAC,Ihler-SSP}, an $O(N)$ greedy algorithm (where $N$ is the number of 
components before reduction) is used to trade off representation 
size with approximation error under various measures.

\subsection{LDLC decoder}

The LDLC decoding algorithm is also described as a message-passing algorithm defined on a factor 
graph\cite{BibDB:FactorGraph}, whose factors represent the
information and constraints on $x$ arising from our knowledge of $y$ and the fact that $b$ is integer-valued.
Here, we rewrite the LDLC decoder update rules in the more standard graphical models notation. The factor graph used is a bipartite graph with variables nodes $\{x_i\}$, representing each element of the vector $x$,
and factor nodes $\{f_i, g_s\}$ corresponding to functions
\small
\begin{align*}
 f_i(x_i) &= \mathcal{N}(x_i; y_i, \sigma^2)\,, &
 g_s(x_{s}) &= \begin{cases} 1 & H_s x \in \Z \\ 0 & \mathrm{otherwise} \end{cases}\,,
\end{align*}
\normalsize
where $H_s$ is the $s^\mathrm{th}$ row of the decoding matrix $H$.
Each variable node $x_i$ is connected to those factors for which it is an argument; since $H$ is sparse,
$H_s$ has few non-zero entries, making the resulting factor graph sparse as well.  
Notice that unlike the construction of~\cite{LDLC_Sommer}, this formulation does not require that
$H$ be square, and it may have arbitrary entries, rather than being restricted to a Latin square
construction.  Sparsity is preferred, both for computational efficiency and because belief propagation
is typically more well behaved on sparse systems with sufficiently long cycles~\cite{BibDB:FactorGraph}.
We can now directly
derive the belief propagation update equations as Gaussian mixture distributions, corresponding to 
an instance of the NBP algorithm.  We suppress the iteration number $\tau$ to reduce clutter.

\newcommand{\ls}{{l}}
\newcommand{\lv}{{\mathbf l}}

\emph{Variable to factor messages.} Suppose that our factor to variable messages $M_{si}(x_i)$ are each
described by a Gaussian mixture distribution, which we will write in both the moment and information form:
\small
\begin{equation} \label{ftovInitLDLC}
M_{si}(x_i) = \sum_\ls w_{si}^\ls \N(x_i \,;\, m_{si}^\ls, \nu_{si}^\ls) = \sum_\ls  w_{si}^\ls \N^{-1}(x_i \,;\, \beta_{si}^\ls, \alpha_{si}^\ls)\,.
\end{equation}
\normalsize
Then, the variable to factor message $M_{is}(x_s)$ is given by
\small
\begin{equation} \label{vtofLDLC}
M_{is}(x_s) = \sum_l w_{is}^l \N(x_s \,;\, m_{is}^l, \nu_{is}^l) = \sum_l  w_{is}^\lv \N^{-1}(x_s \,;\, \beta_{is}^l, \alpha_{is}^l)\,,
\end{equation}
\normalsize
where $\lv$ refers to a vector of indices $[\ls_s]$ for each neighbor $s$,
\small
\begin{align} \label{vtofParamLDLC}
\alpha_{is}^{\vl} &= \sigma^{-2} + \sum_{t\in\nbr_i\setminus s} \alpha_{ti}^{\ls_t}\,,   \ \ \ \beta_{it}^\lv = y_i \sigma^{-2} + \sum_{t\in\nbr_i\setminus s} \beta_{ti}^{\ls_s}\,,   \\
w_{it}^\lv &= \frac{\N(x^*; y_i, \sigma^2) \prod w_{si}^{\ls_s} \N^{-1}(x^*; \beta_{si}^{\ls_s},\alpha_{si}^{\ls_s})}{\N^{-1}(x^*;\beta_{it}^\lv,\alpha_{it}^\lv)}\,.
\nonumber \end{align}
\normalsize
The moment parameters are then given by $\nu_{it}^\lv = (\alpha_{it}^\lv)^{-1}$, $m_{it}^\lv = \beta_{it}^\lv (\alpha_{it}^\lv)^{-1}$.
The value $x^*$ is any arbitrarily chosen point, often taken to be the mean $m_{it}^\lv$ for numerical reasons. 

\emph{Factor to variable messages.} Assume that the incoming messages are of the form~\eqref{vtofLDLC},
and note that the factor $g_s(\cdot)$ can be rewritten in a summation form, $g_s(x_s) = \sum_{b_s} \delta(H_s x = b_s)$, which
includes all possible integer values $b_s$.  
If we condition on the value of both the integer $b_s$ and the indices of the incoming messages,
again formed into a vector $\lv = [\ls_j]$ with an element for each variable $j$,
we can see that $g_s$ enforces the linear equality $H_{si}x_i = b_s - \sum H_{sj} x_j$.
Using standard Gaussian identities in the moment parameterization and summing over all possible $b_s \in \Z$ and $\lv$, we obtain
\small
\begin{align}\label{ftovLDLC}
M_{si}(x_i) = \sum_{b_s} \sum_\lv w_{si}^\lv \N(x_i \,;\, m_{si}^\lv, \nu_{si}^\lv) = \nonumber \\\sum_{b_s} \sum_\lv  w_{si}^\lv \N^{-1}(x_i \,;\, \beta_{si}^\lv, \alpha_{si}^\lv)\,,
\end{align}
\normalsize
where 
\small
\begin{align} \label{ftovParamLDLC}
\nu_{si}^\lv &= H_{si}^{-2} (\sum_{j\in\nbr_s\setminus i} H_{js}^2 \nu_{js}^{\ls_j})\,,  & \nonumber \\
m_{si}^\lv &= H_{si}^{-1} (-b_s + \sum_{j\in\nbr_s\setminus i} H_{js}m_{js}^{\ls_j})\,, & w_{si}^\lv = \prod_{j\in\nbr_s\setminus i} w_{js}^{\ls_j}\,, 
\end{align}
\normalsize
and the information parameters are given by $\alpha_{si}^\lv = (\nu_{si}^\lv)^{-1}$ and $\beta_{si}^\lv = m_{si}^\lv (\nu_{si}^\lv)^{-1}$.

Notice that~\eqref{ftovLDLC} matches the initial assumption of a Gaussian mixture given in~\eqref{ftovInitLDLC}.  At each iteration, the exact messages
remain mixtures of Gaussians, and the algorithm iteslf corresponds to an instance of NBP.  As in any NBP implementation,
we also see that the number of components is increasing at each iteration and must eventually approximate the messages using some
finite number of components.  To date the work on LDLC decoders has focused on deterministic approximations~\cite{LDLC_Sommer,LDLC_Brian,LDLC_Yona,LDLC_Brian2}, often greedy
in nature.
However, the existing literature on NBP contains a large number of deterministic and stochastic approximation algorithms~\cite{NBP2,Ihler-JSAC,KDE,briers05,rudoy07}.
These algorithms can use spatial data structures such as KD-Trees to improve efficiency and avoid the pitfalls that come with greedy optimization.

\emph{Estimating the codewords.}
The original codeword $x$ can be estimated using its belief, an approximation to its marginal distribution
given the constraints and observations:
\small
\begin{equation}
B_i(x_i) = f_i(x_i) \prod_{s \in \nbr_i} M_{si}(x_i)\,.
\end{equation}
\normalsize
The value of each $x_i$ can then be estimated as either the mean or mode of the belief, e.g.,
$x_i^* = \arg \max B_{i}(x_i)$, and the integer-valued information vector estimated as
$b^* = \mathop{\mathrm{round}}(Hx^*)$. 

%  To do so,
%the variable nodes compute $x_i^* = \arg \max M^\tau_{i}(x_i) $, where
%$M^\tau_{i}(x_i) = f_i(x_i) \prod_{t\in \nbr_i} M^\tau_{ti}(x_i)$.  The estimated integer-valued
%information is then $b^* = \mathop{\mathrm{round}}(Hx^*)$.

%%
%For a node with $d$ neighbors and incoming messages of $n$ components each, this produces a mixture model with $n^d$
%components.

%The variable-check node messages are gaussian mixtures of the type
%\[ q_{is}(z) \propto \mathcal{N}(z; m_{is},v_{is})\,, \]
%where the following update rule computes the products of every element from the incoming messages which contains mixtures.
%\begin{equation*}
%\frac{1}{v_{is}} = \frac{1}{\sigma^2} + \displaystyle\sum_{t \in \nbr_i \backslash s} \frac{1}{v_{ti}}\,, \ \ \ \ \ \ \frac{m_{is}}{v_{is}} = \frac{y_i}{\sigma^2} + \displaystyle\sum_{t \in \nbr_i \backslash s} \frac{m_{ti}}{v_{ti}}\,. \label{v2c}
%\end{equation*}
%Note that assuming node $i$ has $k$ neighbors, and each one sends him a mixture of size $p$ elements, the output mixture $q_{ij}$ will contain $N = k^p$ elements.
%Finally, the variable nodes compute $ x_i = \mathop{\max} q_{i*}(z) $, where $q_{i*}(z)$ is the product of all incoming
%mixtures in node $i$, and the result is $\mathop{\mathrm{round}}(Hx)$.

\section{A Pairwise Construction of the LDLC\ decoder}
\label{s-GaBP}
Before introducing our novel lattice code construction, we demonstrate that the LDLC decoder can be
equivalently constructed using a \emph{pairwise} graphical model.  This construction will have 
important consequences when relating the LDLC decoder to Gaussian belief propagation (Section~\ref{NBPandGaBP}) and
understanding convergence properties (Section~\ref{s-new-const}).

%Prior to introducing our novel construction, we would like to review again the LDLC decoder algorithm and show it is an instance of the NBP algorithm. Following we provide a cleaner and simpler derivation of the LDLC decoder algorithm. This derivation will be later used in Section IV\ for deriving additional convergence results.

\begin{thm}
\label{LDLCisNBP}
The LDLC decoder algorithm is an instance of the NBP algorithm executed on the following pairwise graphical model. Denote the number LDLC  variable nodes as $n$ and the number of check nodes as $k$\footnote{Our construction extends the square parity check matrix assumption to  the general case.}. We construct a new graphical model with $n+k$ variables, $X = (x_1, \cdots, x_{n+k}$) as follows. To match the LDLC notation we use the index letters $i,j,..$ to denote variables $1,...,n$ and the letters $s,t,...$ to denote new variables $n+1,...,n+k$ which will take the place of the check node factors in the original formulation. We further define the self and edge potentials: \small
\begin{align}
\psi_i(x_i) \propto \N(x_i; y_i, \sigma^2)\,, &
\ \ \  \psi_s(x_{s}) \triangleq \sum_{b_{s}=-\infty}^{\infty}\N(x_s;b_{s},0)\,, \nonumber \\ \psi_{i,s}(x_i,x_s) &\triangleq \exp(-x_iH_{is}x_s)\,. \label{LDLC_pots}
\end{align}
\normalsize
\end{thm}

\begin{proof} The proof is constructed by substituting the edge and self potentials \eqref{LDLC_pots} into the belief propagation update rules. Since we are using a pairwise graphical model, we do not have two update rules from variable to factors and from factors to variables. However, to recover the LDLC update rules, we make the artificial distinction between the variable and factor nodes, where the nodes $x_i$ will be shown to be related to the variable nodes in the LDLC decoder, and the nodes $x_s$ will be shown to be related to the factor nodes in the LDLC decoder.

\paragraph{LDLC variable to factor nodes}
We start with the integral-product rule computed in the $x_i$ nodes: 
\small
\[ 
M^{}_{is}(x_s) =\int \limits_{x_i}  \psi(x_i,x_s)\psi_i(x_i) \prod_{t\in \nbr_i \setminus s} M_{ti}(x_i)dx_{i}\, \]
\normalsize
The product of a mixture of Gaussians  \small $\prod \limits_{t\in \nbr_i \setminus s} M_{ti}(x_i)$ \normalsize is itself a mixture of Gaussians, where each component in the output mixture is the product of a single Gaussians selected from each input mixture $M_{ti}(x_i)$. 
\begin{lem}[Gaussian product]\label{lem:prod} \cite[Claim 10]{phd-thesis}, \cite[Claim 2]{LDLC_Sommer} Given $p$ Gaussians $\N(m_1,v_1), \cdots, \N(m_p,v_p)$ their product is proportional to a Gaussian $\N(\bar{m},\bar{v})$ with
\small
\begin{align*} \bar{v}^{-1} =  \sum_{i=1}^p \frac{1}{v_{i}} = \sum_{i=1}^p\alpha_i\, & & 
\bar{m} =(\sum_{i=1}^pm_i/v_i)\bar{v}=\sum_{i=1}^p\beta_i\bar{v}
\end{align*} \normalsize
\end{lem}
\begin{proof}
Is given in \cite[Claim 10]{phd-thesis}.
\end{proof}
Using the  Gaussian product 
lemma the $l_s$ mixture component in the message from variable node $i$ to factor node $s$ is a single Gaussian given by
\scriptsize
\[ M^{l_s}_{is}(x_s) =\int \limits_{x_i}  \psi_{is}(x_i,x_s)\big(\psi_i(x_i) \prod_{t\in \nbr_i \setminus s} M^\tau_{ti}(x_i)\Big)dx_{i}= \]
\vspace{-2mm}
\[ \lint{x_i}  \psi_{is}(x_i,x_s)\big(\psi_i(x_i) \exp\{-\tfrac{1}{2} x_i^2(\sum_{t\in \nbr_i \setminus s}\alpha_{ti}^{l_{s}})+x_i(\!\!\sum_{t \in \nbr_i \backslash s} \beta_{ti}^{l_s}\,)\}\big) dx_{i}  = \]
\vspace{-2mm}
 \[  \lint{x_i}  \psi_{is}(x_i,x_s)\Big(\exp(-\tfrac{1}{2} x_{i}^2\sigma^{-2}+x_iy_i\sigma^{-2})\cdot  \] \vspace{-3mm}\[ \cdot\lexp{x_i^2(\sum_{t\in \nbr_i \setminus s}\alpha_{ti}^{l_{s}})+x_i(  \sum_{t \in \nbr_i \backslash s} \beta_{ti}^{l_s}\,)}\Big)dx_{i}  =\]
\vspace{-2mm}
 \[  \lint{x_i}  \psi_{is}(x_i,x_s)\Big( \lexp{x_i^2(\sigma^{-2} +\hspace{-3mm} \sum_{t\in \nbr_i \setminus s}\!\!\!\alpha_{ti}^{l_{s}})+x_{i}(y_i\sigma^{-2}+  \hspace{-3mm}\sum_{t \in \nbr_i \backslash s} \beta_{ti}^{l_s}\,)\,}\Big)dx_{i} \,. \]
\normalsize
We got a formulation which is equivalent to LDLC variable nodes update rule given in \eqref{vtofParamLDLC}.
Now we use the following lemma for computing the integral:

\begin{lem}[Gaussian integral]
\label{lem:int}
Given a (one dimensional)\ Gaussian $\phi_i(x_i) \propto \N(x_i;m,v) $, the
integral $\int \limits_{x_i}  \psi_{i,s}(x_i,x_s)\phi_i(x_i)dx_{i}$, where is a (two dimensional) Gaussian $\psi_{i,s}(x_i,x_s) \triangleq \exp(-\tfrac{}{}x_iH_{is}x_s)$\, is proportional to a (one dimensional)\ Gaussian $\N^{{-1}}(H_{is}m, H_{is}^2 v).$  \end{lem}
\begin{proof}
\small
\[     \int \limits_{x_{i}} \psi_{ij}(x_i,x_j) \phi_{i}(x_i)dx_{i}\ \ \  \] \vspace{-4mm} \[ \propto   \int \limits_{x_{i}} \exp{(-x_iH_{is}x_s)}\lexp{(x_i- m)^2/v}dx_{i}= \]  \vspace{-4mm}
\[ =\lint{x_{i}}\exp{\Big((-\tfrac{1}{2}x_i^2/v)+(m/v-H_{is}x_{s})x_{i}\Big)}dx_{i} \ \ \ \]\vspace{-4mm} \[ \propto \ \ \ \exp{((m/v-H_{is}x_s)^{2}/(-\tfrac{2}{v}))}\,, \]
\ns
where the last transition was obtained by using the Gaussian
integral:
%\end{proof}
%\end{proof}
\BE\label{eq_GaussianIntegral}
\int \limits_{-\infty}^{\infty}\exp{(-ax^{2}+bx)}dx=\sqrt{\pi/a}\exp{(b^{2}/4a)}. \nonumber \EE
\[  \exp{((m/v-H_{is}x_s)^{2}/(-\tfrac{2}{v}))} = \lexp{(v(m/v-H_{is}x_s)^{2})}= \] 
\[ =\lexp{(H_{is}^2v)x_s^2+(H_{is}m)x_s-\tfrac{1}{2}v(m/v)^2} \]\[\propto \lexp{(H_{is}^2 v)x_s^2+(H_{is}m)x_s}\,.\]
\end{proof}
Using the results of Lemma \ref{lem:int} we get
that the sent message between variable node to a factor node is a mixture of Gaussians, where each Gaussian component $k$ is given by 
\small
\[
M^{\lv}_{is}(x_s) = \N^{-1}(x_s ; H_{is}m_{is}^{l_s},H^2_{is}v_{is}^{l_s} ) \,. \]
\normalsize
 Note that in the LDLC terminology the integral operation as defined in Lemma \ref{lem:int} is called stretching. In the LDLC algorithm, the stretching is computed by the factor node as it receives the message from the variable node. In NBP, the integral operation is computed at the variable nodes.
\paragraph*{LDLC\ Factors to variable nodes}We start again with the BP\ integral-product rule and handle the $x_s$ variables
computed at the factor nodes. \small \vspace{-2mm}\[
M^{}_{si}(x_i) = \int \limits_{x_s}  \psi_{is}(x_i,x_s)\psi_s(x_s) \prod \limits_{j\in \nbr_s \setminus i} M_{js}(x_j)\,dx_{s}. \]\normalsize
Note that the product $\prod \limits_{j\in \nbr_s \setminus i} M^\tau_{js}(x_j)$ , is a mixture of Gaussians, where the $k$-th component is computed by selecting a single Gaussian from each message $M^\tau_{js}$ from the set $j \in \nbr_s \setminus i$ and applying the  product lemma (Lemma \ref{lem:prod}). We get 
\small
\BE \lint{x_{s}} \psi_{is}(x_i,x_s)\Big(\psi_s(x_s) \exp \{-\tfrac{1}{2}{x_s^2(\sum_{k \in \nbr_s \backslash i}H^2_{ks}v_{ks}^{l_i}})+\nonumber \EE \vspace{-3mm} \BE + x_s(\sum_{k \in \nbr_s \setminus i} H_{ks} m_{ks}^{l_i} )\,\}\Big)dx_{s}  \label{eq:bs}
\EE
\normalsize
We continue by computing the product with the self potential $\psi_s(x_s)$ to get
\scriptsize
\BE = \lint{x_{s}}  \psi_{is}(x_i,x_s)\big(\!\!\!\!\sum_{b_{s}=-\infty}^{\infty}\!\!\!\!\!\exp(b_{s}x_s) \exp\{-\tfrac{1}{2}x_s^2(\!\!\!\!\sum_{k \in \nbr_s \backslash i}\!\!\!\!\!H^2_{ks}v_{ks}^{l_{i}})+\nonumber \EE \vspace{-6mm}\BE+x_s(  \!\!\!\!\sum_{k \in \nbr_s \setminus i} \!\!\!\!H_{ks} m_{ks}^{l_i} )\,\}\big)dx_{s}= \nonumber \EE
\vspace{-3mm}
\BE =\!\!\!\!\!\sum_{b_{s}=-\infty}^{\infty} \lint{x_{s}}  \psi_{is}(x_i,x_s)\big(\!\exp(b_{s}x_s) \exp\{-\tfrac{1}{2}x_s^2(\!\!\!\sum_{k \in \nbr_s \backslash i}\!\!\!\!H^2_{ks}v_{ks}^{l_i})+\nonumber \EE \vspace{-6mm}\BE +x_s(\!\!\!  \sum_{k \in \nbr_s \setminus i} \!\!\!\!\!H_{ks} m_{ks}^{l_i} )\,\}\big)dx_{s}=  \nonumber
\EE\vspace{-4mm}
\vspace{-2mm}
\BE =\!\!\!\!\sum_{b_{s}=-\infty}^{\infty}\lint{x_{s}}  \psi_{is}(x_i,x_s)\big( \exp\{-\tfrac{1}{2} x_s^2(\!\!\!\!\sum_{k \in \nbr_s \backslash i}\!\!\!\!\!H^2_{ks}v_{ks}^{l_i})+\nonumber \EE \vspace{-4mm} \BE\ \!\!x_s(b_{s}+ \!\!\!\!\!\sum_{k \in \nbr_s \setminus i}\!\!\!\! H_{ks} m_{ks}^{l_i})\}\big)dx_{s}=  \nonumber
\EE\vspace{-4mm}
\hspace{-2pt}
\BE =\!\!\sum_{b_{s}=\infty}^{-\infty} \lint{x_{s}}  \psi_{is}(x_i,x_s)\big( \exp\{-\tfrac{1}{2}x_s^2(\!\!\!\!\sum_{k \in \nbr_s \backslash i}\!\!\!\!\!H^2_{ks}v_{ks}^{l_i})+ \nonumber \EE \vspace{-4mm} \BE +x_s(-b_{s}+\!\!\!\!  \sum_{k \in \nbr_s \setminus i}\!\!\!\!\!\! H_{ks} m_{ks}^{l_i} )\,\}\big)dx_{s}\,. \nonumber
\EE \vspace{-2mm}
\normalsize
Finally we use Lemma \ref{lem:int} to compute the integral and get
\scriptsize \BE =\sum_{b_{s}=\infty}^{-\infty}  \exp\{-tfrac{1}{2}x_s^2H_{si}^{2}(\!\!\sum_{k \in \nbr_s \backslash i}H^{2}_{ks}v_{ks}^{l_i})^{-1}+\nonumber \EE \vspace{-3mm} \BE +x_sH_{si}(\sum_{k \in \nbr_s \backslash i}H^{2}_{ks}v_{ks}^{l_i})^{-1}(-b_{s}+  \sum_{k \in \nbr_s \setminus i} H_{ks} m_{ks}^{l_i} )\,\}dx_{s} \,. \nonumber
\EE\normalsize
It is easy to verify this formulation is identical to the LDLC update rules \eqref{ftovParamLDLC}. % \begin{lem}
\end{proof}

\section{Using Sparse Generator Matrices}
\label{s-novel}

We propose a new family of LDLC\ codes where the generator matrix $G$ is sparse, in contrast to the original LDLC codes where the parity check matrix $H$ is sparse. Table \ref{tab1} outlines the properties of our proposed decoder.
Our decoder is designed to be more efficient than the original LDLC decoder, since as we will soon show, both encoding,   initialization and final operations are more efficient in the NBP decoder.
We are currently in the process of fully evaluating our decoder performance relative to the LDLC decoder. Initial results are reported in Section VI.  
%Further discussion on the applicability of GaBP to non-square
%matrices is found in \cite{ISIT1,ISIT2}.

\subsection{The NBP decoder}
We use an undirected
bipartite graph, with variables nodes $\{b_i\}$, representing each element of the vector $b$, and observation nodes $\{z_i\}$ for each element of the observation vector $y$.
We define the self potentials $\psi_i(z_i)$ and $\psi_s(b_s)$ as follows: \small
\begin{align}
 \psi_i(z_i) & \propto \N(z_i;y_i, \sigma^2)\,, &
 \psi_s(b_{s}) &= \begin{cases} 1 & b_{s} \in \Z \\ 0 & \mathrm{otherwise} \end{cases}\,, \label{eq:delta}
\end{align}
\normalsize
and the edge potentials:
\small
\[
 \psi_{i,s}(z_i, b_s) \triangleq \exp(-z_iG_{is}b_s)\,. 
\]
\normalsize
Each variable node $b_s$ is connected to the observation nodes as defined by the encoding matrix $G. $ Since $G$ is sparse, the resulting bipartite graph sparse as well.
\textcolor{red}{} As with LDPC decoders~\cite{BibDB:McElieceMacKayCheng}, the belief propagation or sum-product algorithm~\cite{BibDB:BookPearl,BibDB:FactorGraph}
provides a powerful approximate decoding scheme.
%  The purpose of belief propagation is to estimate the posterior marginal distributions $p(x_i|y)$

For computing the MAP assignment of the transmitted vector $b$ using non-parametric belief propagation we perform
the following relaxation, which is one of the main novel contributions of this paper.
Recall that in the original problem, $b$ are only allowed to be integers.  We relax the function $\psi_s(x_s)$ from a delta function to a mixture of  Gaussians centered around integers.\small \[ \psi^{relax}_s(b_{s}) \propto \sum_{i \in \Z} \mathcal{N}(i, v) \,.\]  \normalsize
\begin{figure}
\centering{
\includegraphics[scale=0.4]{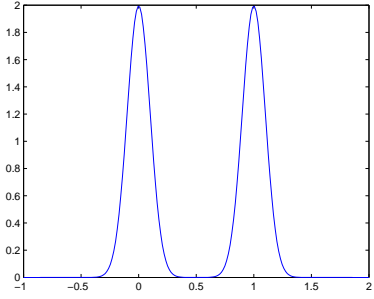}
\label{fig:self_pot}
\caption{The approximating function $g^{relax}_s(x)$ for the binary case.}
}
\end{figure}
The variance parameter $v$ controls the approximation quality, as $v \rightarrow 0$ the approximation quality is higher. Figure 2 plots an example relaxation of $\psi_i(b_s)$ in the binary case.
We have defined the self and edge potentials which are the input the to the NBP algorithm. Now it is possible to run the NBP\ algorithm using \eqref{eq:BP} and get an approximate MAP solution to \eqref{eq:ls}.
The derivation of the NBP decoder update rules is similar to the one done for the LDLC decoder, thus omitted. However, there are several important differences that should be addressed. We start by analyzing the algorithm efficiency.

\begin{table}[h!]
\centering{
\begin{tabular}{|c|c|c|}\hline
Algorithm & LDLC & NBP \\\hline\hline
Initialization operation &$G = H^{-1}$& None\\
Initialization cost &$O(n^3)$& -\\\hline
Encoding operation & $Gb$ & $Gb$ \\
Encoding cost & $O(n^2)$ & $O(nd), \ \ \ d\ll  n$ \\\hline
%Cost per node per iteration & $O(q\log(q)d)$ & $O(qd)$ \\ \hline 
Post run operation & $Hx$ & None \\
Post run cost & $O(nd)$ & - \\\hline \hline
\end{tabular}
\caption{Comparison of LDLC decoder vs. NBP decoder }
\label{tab1}
}
\end{table}
We assume that the input to our decoder is the sparse matrix $G$, there is no need in computing the encoding matrix $G=H^{-1}$ as done in the LDLC decoder. Naively this initialization takes $O(n^3)$ cost.
The encoding in our scheme is done
as in LDLC by computing the multiplication $Gb$. However, since $G$ is sparse in our case, encoding cost is $O(nd)$ where $d << n$ is the average number of  non-zeros entries on each row.
Encoding in the LDLC method is done in $O(n^2)$ since even if $H$ is sparse, $G$ is typically dense.
After convergence, the LDLC decoder multiplies by the matrix $H$ and rounds the result to get $b$. This operation costs $O(nd)$ where $d$ is the average number of non-zero entries in $H$. In contrast, in the NBP decoder, $b$ is computed directly in the variable nodes.

Besides of efficiency, there are several inherent differences between the two algorithms. Summary of the differences is given in Table \ref{tab:diff}. \ We use a standard formulation of BP
using pairwise potentials form, which means there is a single update rule, and not two update rules from left to right and right to left. We have shown that the convolution operation in the LDLC decoder relates to product step of the BP\ algorithm. The stretch/unstrech operations in the LDLC decoder are implemented using the integral step of the BP algorithm. The periodic extension operation in the LDLC decoder is incorporated into our decoder algorithm using the self potentials.\begin{table}
\scriptsize
\centering{
\begin{tabular}{|c|c|c|}\hline
Algorithm &\ LDLC decoder &\ NBP decoder\\ \hline \hline
Update rules & Two& One\\\hline
Sparsity assumption & Decoding mat. H& Encoding mat. G\\\hline
Algorithm derivation & Custom & Standard NBP \\ \hline
Graphical model & Factor graph & Pairwise potentials\\ \hline
Related Operations &\ Stretch/Unstretch &\ Integral \\ 
&\ Convolution &\ product\\
&\ periodic extension &\ product \\ \hline
\end{tabular} 
\caption{Inheret differences between LDLC and NBP decoders}
\label{tab:diff}
}
\end{table}
\normalsize
\vspace{-3mm}
\subsection{The relation of the NBP decoder to GaBP} \label{NBPandGaBP}
%\subsection{NBP decoder}
In this section we show that simplified version of the NBP decoder coincides with the GaBP algorithm. 
The simplified version is obtained, when instead of using our proposed Gaussian mixture prior, we initialize the NBP algorithm with a prior composed of a single Gaussian.\begin{thm} \label{NBPisGaBP}
By initializing  $\psi_s(b_{s}) \sim \N(0,1)$ to be a (single) Gaussian the NBP\ decoder update rules are identical to update rules of the GaBP algorithm.
\end{thm}
\begin{lem}
By initializing  $\psi_s(x_{s})$ to be a (single) Gaussian the messages of the NBP decoder are single Gaussians.\end{lem}
\begin{proof}
Assume both the self potentials $\psi_s(b_{s}), \psi_{i}(z_i) $ are initialized to a single Gaussian, every message of the NBP decoder algorithm will remain a Gaussian. This is because the product \eqref{eq:GaBP1} of single Gaussians is a single Gaussian, the integral and   \eqref{eq:GaBP2} of single Gaussians produce a single Gaussian as well.
\end{proof}
%The proof to theorem 4 is given at the technical report version of this %paper \cite{TechReport}.
 Now we are able to prove Theorem 4:\\ \begin{proof}We start writing the update rules of the variable nodes. We initialize the self potentials of the variable nodes $\psi_i(z_i) = \N(z_i;y_i, \sigma^2)\,$,
   Now we substitute, using the product lemma and Lemma \ref{lem:int}.
 \scriptsize
 \[ M_{is}(b_{s}) = \int \limits_{z_i} \psi_{i,s}(z_i,b_s)\Big( \psi_i(z_i)\prod \limits_{t\in \nbr_i \setminus s} M_{ti}(z_i)\Big)dz_i = \] \vspace{-3mm}
 \[\lint{z_i} \psi_{i,s}(z_i,b_s)\big( \exp (-\tfrac{1}{2}z_i^{2}\sigma^{-2}+y_iz_i\sigma^{-2})\!\!\!\prod \limits_{t\in \nbr_i \setminus s}\!\!\!  \exp (-\tfrac{1}{2}z_i^2\alpha_{ti}+z_i\beta_{ti})\big)dz_{i} \]\vspace{-2mm}
 \[\lint{z_i} \psi_{i,s}(z_i,b_s)\big( \exp(-\tfrac{1}{2}z_i^2(\sigma^{-2} + \sum_{t \in \nbr_i \setminus s} \alpha_{ti})+z_{i}(\sigma^{-2}y_i+\!\!\!\sum_{t \in \nbr_i \setminus s} \beta_{ti})\big)dz_{i} = \]\vspace{-2mm}
  \[\propto \exp\Big(-\tfrac{1}{2}z_i^2G^2_{is}(\sigma^{-2} +\sum_{t \in \nbr_i \setminus s} \alpha_{ti})^{-1}+\] \vspace{-2mm} \[z_{i}G_{is}(\sigma^{-2} +\!\!\sum_{t \in \nbr_i \setminus s} \alpha_{ti})^{-1}(\sigma^{-2}y_i+\sum_{t \in \nbr_i \setminus s} \beta_{ti})\Big)\]
 \normalsize
 Now we get GaBP update rules by substituting $J_{ii} \triangleq \sigma^{-2}, J_{is} \triangleq G_{is}, h_s \triangleq \sigma^{-2}y_i:$
\small
 \begin{align*}
 {\alpha_{is} = -J_{is}^2 \alpha_{i \setminus s}^{-1} = -J_{is}^2 ( J_{ii} + \sum \limits_{{t} \in \nbr_i \setminus s} \alpha_{ti})^{-1} },
 & &\nonumber \\ {\beta_{is}=-J_{is} \alpha_{i \setminus s}^{-1} \beta_{i \setminus s } =-J_{is} \Big(\alpha_{i \setminus s}^{-1}  (h_i + \sum \limits_{t \in \nbr_i \setminus s} \beta_{ti})\Big)}\,.
 \end{align*}
 \normalsize
 We continue expanding \[ M_{si}(z_{i}) = \int \limits_{b_s} \psi_{i,s}(z_i,b_s)\Big( \psi_s(b_{s})\prod \limits_{k \in \nbr_s \backslash i} M^\tau_{ks}(b_{s})\Big)db_{s} \]  Similarly using the initializations \\$\psi_s(b_{s}) = \lexp{b_s^2} ,\ \psi_{i,s}(z_i, b_s) \triangleq \exp(-z_iG_{is}b_s). $ 
 \scriptsize
 \[ \lint{b_s} \psi_{i,s}(z_i,b_s)\Big( \lexp{b_s^2}\prod \limits_{k \in \nbr_s \backslash i}  \exp(-\tfrac{1}{2}b_s^{2}\alpha_{is}+b_{s}\beta_{ks})\Big)db_{s} = \] \vspace{-2mm}
 \[ \lint{b_s} \psi_{i,s}(z_i,b_s)\Big(\lexp{b_s^2(1+\!\!\!\sum \limits_{k \in \nbr_s \backslash i}\alpha_{is})+b_s(\sum \limits_{k \in \nbr_s \backslash i}\beta_{ks})^{}}\Big)db_{s}=  \]\vspace{-2mm}
  \[\lexp{b_s^2G_{is}^2(1+\!\!\!\!\sum \limits_{k \in \nbr_s \backslash i}\!\!\!\alpha_{is})^{-1}+b_sG_{is}(1+\sum \limits_{k \in \nbr_s \backslash i}\alpha_{is})^{-1}(\!\!\!\!\sum \limits_{k \in \nbr_s \backslash i}\beta_{ks})} \]
 \normalsize
  Now we get GaBP update rules by substituting $J_{ii} \triangleq1,\\  J_{si} \triangleq G_{is}, h_i \triangleq 0:$
\small
 \begin{align*}
 {\alpha_{si} = -J_{si}^2 \alpha_{s \setminus i}^{-1} = -J_{si}^2 ( J_{ii} + \sum \limits_{{k} \in \nbr s\setminus i} \alpha_{is} )^{-1} },\nonumber \\
 {\beta_{si}=-J_{si} \alpha_{s \setminus i}^{-1} \beta_{s \setminus i } =-J_{si} \Big(\alpha_{s \setminus i}^{-1}  (h_i + \sum \limits_{{k} \in \nbr s\setminus i} \beta_{ks})\Big)}\,.
 \end{align*}
\vspace{-7mm}
\normalsize

\end{proof}
Tying together the results, in the case of a single Gaussian self potential, the NBP decoder is initialized using the following inverse covariance matrix: 
 \[
J \triangleq 
\begin{pmatrix}
I & G\\
G^T &\ \diag(\sigma^{-2}) 
\end{pmatrix}
\]
%\end{proof}
We have shown that a simpler version of the NBP decoder, when the self potentials are initialized to be single Gaussians boils down to GaBP algorithm. It is known \cite{ISIT1} that the GaBP\ algorithm  solves the following least square problem $ \min_{b \in \R^{n}} \|Gb-y\| $ assuming a Gaussian prior on $b$, $p(b) \sim \N(0,1)$, we get the MMSE solution
$ b^* =(G^TG)^{-1}G^Ty. $
Note the relation to \eqref{eq:ls}. The difference is that we relax the LDLC decoder assumption that $b \in \Z^n$, with  $b \in \R^n$.  

Getting back to the NBP decoder, Figure 2 compares the two different priors used, in the NBP decoder and in the GaBP algorithm, for the bipolar case. \normalsize It is clear that the Gaussian prior assumption on $b$ is not accurate enough. In the NBP decoder, we relax the delta function \eqref{eq:delta} to a Gaussian mixture prior composed of mixtures centered around Integers. Overall, the NBP decoder algorithm can be thought of as an extension of the GaBP algorithm with more accurate priors.
\vspace{-3mm}
\begin{figure}[h!]
\centering{
%\hspace{-4mm}
\includegraphics[bb=38 182 577 610,scale=0.30,clip]{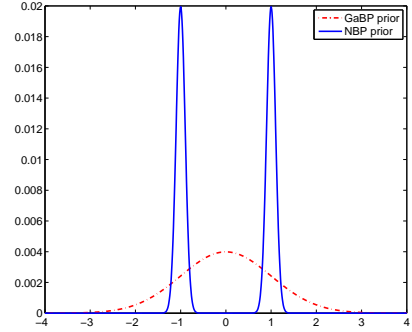}
\label{fig:self_pot}
%\vspace{-2mm}
\caption{Comparing GaBP prior to the prior we use in the NBP\ decoder for the bipolar case $(b \in \{-1,1\}$).}
}
\end{figure}

\vspace{-4.5mm}
\section{Convergence analysis}
\label{s-new-const}
%\subsection{New convergence results}
The behavior of the belief propagation algorithm has been extensively studied in the literature, resulting
in sufficient conditions for convergence in the discrete case~\cite{ihler05b} and in jointly
Gaussian models~\cite{BibDB:mjw_walksum_jmlr06}.  However, little is known about the behavior of BP in more general
continuous systems.  The original LDLC paper~\cite{LDLC_Sommer} gives some characterization of its convergence
properties under several simplifying assumptions.  Relaxing some of these assumptions and using our
pairwise factor formulation, we show that the conditions for GaBP convergence can also be applied to yield 
new convergence properties for the LDLC decoder.

The most important assumption made in the LDLC convergence analysis~\cite{LDLC_Sommer} is that the system 
converges to a set of ``consistent'' Gaussians; specifically, that at all iterations $\tau$ beyond some 
number $\tau_0$, only a \emph{single} integer $b_s$ contributes to the Gaussian mixture.  Notionally,
this corresponds to the idea that the decoded information values themselves are well resolved, and the convergence
being analyzed is with respect to the transmitted bits $x_i$.  Under this (potentially strong) assumption,
sufficient conditions are given for the decoder's convergence.
The authors also assume that $H$ consists of a Latin square in which each row and column contain some permutation
of the scalar values $h_1 \ge \ldots \ge h_d$, up to an arbitrary sign.

Four conditions are given which should all hold to ensure convergence:  \textbf{}
\begin{itemize}
\item{LDLC-I}: $\det(H) = \det(G) = 1$.
\item{LDLC-II}: $\alpha \le 1$, where $\alpha \triangleq \frac{\sum _{i=2}^d h^2_{i}}{h^2_{1}}$.
\item{LDLC-III}: The spectral radius of $\rho(F) < 1$ where $F$ is a $n \times n$ matrix
defined by:
\[
F_{k,l} =
\begin{cases}
\frac{h_{rk}}{h_rl} & \mbox{if $k\ne l$ and there exist a row $r$ of $H$
}\\ & \mbox{for which $|H_{rl}| = h_1$ and $H_{rk} \ne 0$}\\
0 & \mbox{otherwise}  \\
\end{cases}
\]
\item{LDLC-IV}: The spectral radius of $\rho(\tilde{H}) < 1$ where $\tilde{H}$ is derived from H by permuting the rows such that
the $h_1$ elements will be placed on the diagonal, dividing
each row by the appropriate diagonal element ($+h_1$ or $-h_1$),
and then nullifying the diagonal.\end{itemize}

Using our new results we are now able to provide new convergence conditions for  the LDLC decoder. 
\begin{corol}
\label{corol}
 The convergence of the LDLC decoder depends on the properties of the following matrix:
\BE J \triangleq 
\begin{pmatrix}
\0 & H\\
%\diag(\infty) & H\\
H^T &\ \diag(1/\sigma^2)  
\end{pmatrix} \label{LDLCJ} \EE
\end{corol}
\begin{proof}
In Theorem \ref{LDLCisNBP} we have shown an equivalence between the LDLC algorithm to NBP initialized with the following potentials:  
\small
\vspace{-2mm}
\begin{align}
\psi_i(x_i) &\propto \N(x_i; y_i, \sigma^2)\,, \ \ \ \ 
 \psi_s(x_{s}) &\triangleq \sum_{b_{s}=-\infty}^{\infty}\N^{-1}(x_s;b_{s},0)\nonumber \,, \\ 
\vspace{-3mm}
&\psi_{i,s}(x_i,x_s) \triangleq \exp(x_iH_{is}x_s)\,. \label{LDLC_pots}
\end{align} 
\normalsize
We have further discussed the relation between the self potential $\psi_s(x_s)$
and the periodic extension  operation.  We have also shown in Theorem \ref{NBPisGaBP}
that if $\psi_s(x_s)$ is a \emph{single} Gaussian (equivalent to the assumption
of ``consistent'' behavior), the distribution is jointly Gaussian and rather than
NBP (with Gaussian mixture messages), we obtain GaBP (with Gaussian messages).
Convergence of the GaBP algorithm is dependent on
the inverse covariance matrix $J$ and not on the shift vector $h$. 

Now we are
able to construct the appropriate inverse covariance matrix $J$ based on the
pairwise factors given in Theorem \ref{LDLCisNBP}. The matrix $J$ is a $2\times
2$ block matrix, where the check variables $x_s$ are assigned the upper rows and
the original variables are assigned the lower rows.  The entries can be read out
from the quadratic terms of the potentials~\eqref{LDLC_pots}, with the only non-zero entries
corresponding to the pairs $(x_i, x_s)$ and self potentials $(x_i,x_i)$.

%Since $\psi_s(x_s)$ is a
%periodic shift operation, it does not change the variance, does adds zero
%variable to the current value, which means that the precision is infinite.  The
%upper left block matrix of $J,$ which matches $x_s$ is a diagonal matrix
%$\diag(\infty).  $ The inverse covariance given in the edge potentials
%$\psi_{i,s}(x_i,x_s)$ is determined by the entries of $H$. Finally, the self
%potentials $\psi_i(x_i)$ have precision of $1/\sigma^2$ which is initialized in
%a diagonal inverse variance matrix of the lower right block.  
\end{proof}

Based on Corollary \ref{corol} we can characterize the convergence of the LDLC decoder, using the sufficient conditions for convergence of GaBP. Either one of the following two conditions are sufficient for convergence:
\begin{itemize}
\item[][GaBP-I] (\emph{walk-summability} \cite{BibDB:mjw_walksum_jmlr06}) \\$\rho(I-|D^{-1/2}JD^{-1/2}|)<1$ where $D \triangleq \diag(J)$.
\item[][GaBP-II] (\emph{diagonal dominance}  \cite{BibDB:Weiss01Correctness}) $J$ is diagonally dominant (i.e. $|J_{ii}| >= \sum_{j \ne i} |J_{ij}|, \forall i$).
\end{itemize}

A further difficulty arises from the fact that the upper diagonal of \eqref{LDLCJ} is zero, which means that both [GaBP-I,II] fail to hold. There are three possible ways to overcome this.
\begin{enumerate}
\item 
Create an approximation to the original problem by setting the upper  left block matrix of \eqref{LDLCJ} to $\diag(\epsilon)$ where $\epsilon>0$ is a small constant. The accuracy of the approximation grows as $\epsilon$ is smaller. In case either of [GaBP-I,II] holds on the fixed matrix the ``consistent Gaussians'' converge into an approximated solution.
\item  In case a permutation on $J$ \eqref{LDLCJ} exists where either [GaBPI,II] hold for permuted matrix, then the ``consistent Gaussians'' convergence to the correct solution. 
\item Use preconditioning to create a new graphical model where the edge potentials  are determined by the information matrix $HH^T$, $\psi_{i,s}(x_i,x_s) \triangleq \exp(x_{i}\{HH^T\}_{is}x_{s})$ and the self potentials of the $x_i$\ nodes are $\psi_i(x_i) \triangleq \lexp{x_i^2\sigma^{-2}+x_i\{Hy\}_i}$. The proof of the correctness of the above construction is given in \cite{ISIT2}. The benefit of this preconditioning is that the main diagonal of $HH^T$ is surely non zero. If either [GaBP-I,II] holds on $HH^T$ then ``consistent Gaussians'' convergence to the correct solution. However, the matrix $HH^T$ may not be sparse anymore, thus we pay in decoder efficiency.
\end{enumerate}
Overall, we have given two sufficient conditions for convergence, under the ``consistent Gaussian'' assumption for the means and variances of the LDLC decoder. Our conditions are more general because of two reasons. First, we present a single sufficient condition instead of four that have to hold concurrently in the original LDLC work. Second, our convergence analysis does not assume Latin squares,  not even square matrices and does not assume nothing about the sparsity of $H$. This extends the applicability of the LDLC decoder to other types of codes. Note that our convergence analysis relates to the mean and variances of the Gaussian mixture messages. A remaining open problem is the convergence of the amplitudes -- the relative heights of the different consistent Gaussians.
\vspace{-4pt}
\section{Experimental results}
\vspace{-2pt}
 In this section we report preliminary experimental results of our NBP-based\ decoder. Our implementation is general and not restricted to the LDLC domain. Specifically, recent work by Baron \etal\ \cite{CSBP2009} had extensively tested our NBP implementation in the context of the related compressive sensing domain.
Our Matlab code is available on the web on \cite{MatlabGABP}.
 
 We have used a code lengths of $n=100, n=1000$, where the number of non zeros in each row and each column is $d=3$. Unlike LDLC Latin squares which are formed using a generater sequence $h_i$, we have selected the non-zeros entries of the sparse encoding matrix $G$ randomly out of $\{-1,1\}$.
 This construction further optimizes LDLC decoding, since bipolar entries avoids the integral computation (stretch/unstrech operation).\ We have used bipolar signaling, $b \in \{-1,1\}$.  
 We have calculated the maximal noise level $\sigma^2_{max}$ using Poltyrev generalized definition for channel capacity using unrestricted power assumption \cite{Poltyrev}. For bipolar signaling $\sigma^2_{max} = 4\sqrt[n]{\det(G)^2}/2\pi e$. When applied
 to lattices, the generalized capacity implies that there exists a
 lattice $G$ of high enough dimension $n$ that enables transmission
 with arbitrary small error probability, if and only if $\sigma^2 < \sigma^2_{max}$. Figure \ref{performance} plots SER\ (symbol error rate) of the NBP decoder vs. the LDLC decoder for code length $n=100,n=1000$. The $x$-axis represent the distance from capacity in dB as calculated using Poltyrov equation. As can be  seen, our novel NBP decoder has better SER for $n=100$ for all noise levels. For $n=1000$ we have better performance for high noise level, and comparable performance up to 0.3dB from LDLC for low noise levels. We are currently in the process of extending our implementation to support code lengths of up $n=100,000$. Initial performance results are very promising.\ 
\begin{figure}[h!]
%\centering{
\hspace{-2.5cm}
\includegraphics[scale=0.4]{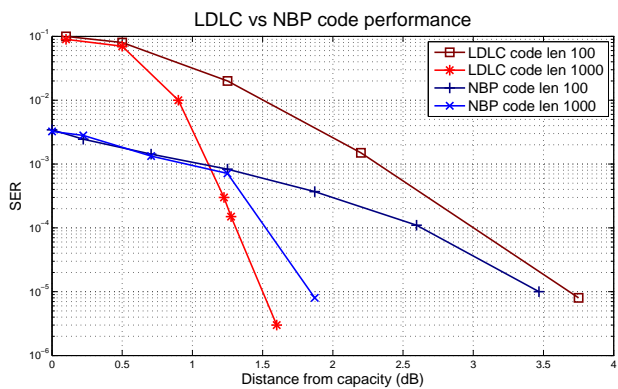}
\caption{NBP vs. LDLC decoder performance}
\label{performance}
%}
\end{figure}
\vspace{-5mm}
\section{Future work and open problems}
We have shown that the LDLC decoder is a variant of the NBP\ algorithm. This allowed us to use current research results from the non-parametric belief propagation domain,  to extend the decoder applicability in several directions. First, we have extended algorithm applicability from Latin squares to  full column rank matrices (possibly non-square). Second, We have extended the LDLC\ convergence analysis, by discovering simpler conditions for convergence. Third, we have presented a new family of LDLC which are based on sparse encoding matrices.

We are currently working on an open source implementation of the NBP\ based decoder, using an undirected graphical model, including a complete comparison of performance to the LDLC decoder. Another area of future work is to examine the practical performance of the efficient Gaussian mixture product sampling algorithms developed in the NBP domain to be applied for LDLC
decoder. As little is known about the convergence of the NBP\ algorithm, we plan to continue examine its convergence in different settings. Finally, we plan to investigate the applicability of the recent convergence fix algorithm \cite{ISIT09-1} for supporting decoding matrices where the sufficient conditions for convergence do not hold.
\ignore{
\subsection{Extending the LDLC decoder to full column rank matrices}
In this section we show that given a full column rank $H$ (possibly non square). 

\begin{lem} The first $k$ entries in the solution of the linear system of equations $Jx = \zeta$ as defined in \eqref{eq:matJ} is identical to the solution of the following system of linear equations
\[ H^THx=H^Ty \]
\end{lem}
\begin{proof} We start by expanding \eqref{eq:matJ} and writing it as a linear system of equations with two variables:  
\[  \left(\begin{array}{cc} \mathbf{0} & H^{T} \\ H & -\tfrac{1}{\sigma^2} I \\ \end{array} \right)\, 
\begin{pmatrix}x \\
z \\
\end{pmatrix}
= 
\begin{pmatrix}0 \\
y \\
\end{pmatrix}\,. \]
 \begin{align*} H^Tz = 0\,, & & &\ Hx-\tfrac{1}{\sigma^2}z=y\,. \end{align*}
\[ H^THx-\tfrac{1}{\sigma^2}\underbrace{H^Tz}_{=0} = H^Ty \]
\[ H^THx=H^Ty \]
\end{proof}
Using Lemma 2 results, we can equivalently check the convergence conditions [GaBP-I,II] on the matrix $H^TH$, instead of the matrix $J$.
Lemma 2 gives us a better understanding of the LDLC decoder, since as we have shown, without the periodic shift operation, the decoder actually computes the ML solution $x = (H^TH)^{-1}H^Ty.$ The periodic shift operation enforces the Integer constraints on $x$.
}%ignore

\ignore{
\section{Experimental results}
\label{s-exp-results}
In this section we tie together the theoretical results by providing an efficient implementation
of the extended LDLC decoder for solving the multiuser detection problem. This construction
could not be easily framed within the original LDLC formulation, since in general the encoding and
decoding matrices are not square.

Consider the multiuser detection problem of $K$ users, each spreading its symbol
onto $N$ chips in a linear AWGN channel.  In this formulation we have $y = Gb + w$,
where $b_i \in \{-1,1\}$ is the binary symbol transmitted by user $i$, $G_{N \times K}$ is the signature matrix,
y is the received vector (the observation), and $w$ is an $N \times 1$ vector
of i.i.d. AWGN with mean $0$ and variance $\sigma^2$.
We are then interested in solving the following least squares problem:
$$ x^* =\arg \mathop{\min }\limits_{b \in \Z^K}||Gb - y||^2 $$

We propose to solve the multiuser detection problem using an LDLC-like decoder. The Matlab source code for this example is
available at \cite{MatlabGABP}. Since, as in LDLC,  the number of Gaussian mixtures grows exponentially during the algorithm,
we have utilized the Kernel Density Estimation Toolbox for Matlab \cite{KDE} to computing approximate products
in the variable nodes.  The toolbox provides an implementation of several efficient, multi-scale approximation
algorithms developed in the context of NBP~\cite{NBP2}.

%We set the limit of number of Gaussians in each mixture to 30.
For our experiments, we construct the following sparse matrix $H$,
in which there are $N = 6$ check nodes and $K= 5$ variable nodes.
%\footnotesize
\[ \footnotesize H_{5 \times 6} \triangleq \left(
  \begin{array}{cccccc}
-0.8  &  0.7 &   0.5 &        0   &      0  & -1\\
  0   &  0   & 1&  -0.5   & 1  & -0.8\\
   -0.5&    1 &         0 &  -0.8 &        0 &   0.5\\
 1&         0 &   0.8 &        0   & 0.5   &      0\\
         0 &  -0.5 &        0   & 1 &   0.8 &        0\\
  \end{array}
\right)  \]
%\normalsize
We compute the corresponding signature matrix $G$ as the pseudo-inverse, $G = H^T(HH^T)^{-1}$, so that $HG = I$.
%In this case, $G_{6 \times 5} =$
%\footnotesize
%\[  =\left(
%  \begin{array}{ccccc}
%   -0.0698  & -0.2489  & -0.2164 &   0.6849  & -0.2672\\
%    0.5979  & -0.5711  &  0.5621 &   0.6277  &  0.1906\\
%    0.2845  &  0.0002  &  0.0073 &   0.4953  & -0.0471\\
%    0.5513  & -0.6836  & -0.0557 &   0.4436  &  0.6075\\
%   -0.3155  &  0.4976  &  0.4210 &  -0.1622  &  0.6098\\
%   -0.3834  & -0.2005  & 0.5702  &  0.1392   &  0.3236\\
%  \end{array}
%\right) \]
%\normalsize

We set the noise level $\sigma^2 = 0.18$, where $\sigma^2 > 1/2\pi e \simeq 0.05$, the noise level where the LDLC
codes reaches the channel capacity \cite[Section II-A]{LDLC_Sommer}. We run the extended LDLC decoder, and set the
variance convergence threshold to $\epsilon = 0.001$. Since $b \in \{-1,1\}$, the periodic shift operation in
Eq.~\eqref{eq:periodic} was implemented by two possible shifts, $-1$ and $1$. Each transmitting user
transmitted $\{-1,1\}$ with equal probability.

Using the above settings, the extended decoder converged after 5 iterations. Figure \ref{msgs} shows the messages
sent between one of the variable nodes, $b_1$, and a check node, $c_1$. The 5 sub-figures correspond to each
iteration of the algorithm, respectively.
In the first
iteration the message is a single Gaussian where its mean is $y_1$, the observation for variable node 1,
where the message variance is $\sigma^2$ \eqref{eq:v2c-init}. On the subsequent rounds, information received from
the check nodes weight mixture mass to the right. Blue triangles represents the centers of the mixture components.

Figure \ref{msgs2} plots the messages sent between check node one and variable node one on the same example.
Because there are two possible integer values: $\{-1,1\}$, the shift operation of Eq.~\eqref{eq:periodic} duplicates
the mixture twice. In the first iteration the mixture is composed of two Gaussians, because the initial
messages from variable to check nodes are a single Gaussian, which is duplicated twice.

\begin{figure}[ht]
\begin{minipage}[b]{0.5\linewidth}
\centering
\includegraphics[bb=118 176 489 622,clip,scale=0.5,height=250pt]{msgs.eps}
%  \begin{picture}(0,0)
%  \put(-242,223){Iter. 1}
%  \put(-122,178){Iter. 2}
%  \put(-122,130){Iter. 3}
%  \put(-122,87){Iter. 4}
%  \put(-122,40){Iter. 5}
%  \end{picture} \\
  \caption{Messages sent between variable $b_1$ one and check node $c_1$.}\label{msgs}
\vspace{0.2cm}
\end{minipage}
%\hspace{0.5cm}
\begin{minipage}[b]{0.5\linewidth}
\centering
\includegraphics[bb=108 203 505 595,scale=0.57,clip,height=250pt]{msgs2.eps}\\
 \begin{picture}(0,0)
 \put(-150,235){Iter. 1}
 \put(-150,185){Iter. 2}
 \put(-150,135){Iter. 3}
 \put(-150,85){Iter. 4}
 \put(-150,35){Iter. 5}
 \end{picture} \\
  \caption{Messages sent between check node $c_1$ and variable node $b_1$.}\label{msgs2}
%\vspace{-0.5cm}
\label{msgs2}
\end{minipage}
\end{figure}

For comparison, we have run an iterative multiuser detection algorithm based on GaBP~\cite{Allerton,ISIT2}. This
iterative algorithm relaxes the problem to a continuous domain, and clips the result using the
$\mathtt{sign()}$ function. We have repeated this experiment
100 times, where each time both the LDLC decoder and the GaBP decoder receive the same input.
In total, the LDLC decoder was able to overcome the noise and detect the transmission correctly
in 88\% of the tests, while the iterative decoder had a success of only 79\%.  This suggests that
the Gaussian mixture messages are better able to capture the imposed constraints and residual uncertainty
in the decoding process.

\section{Conclusion}
\label{s-conc}
We explore the relation between the LDLC decoder and the belief propagation algorithm. We show that the LDLC
decoder is related to two algorithms from machine learning, the non-parametric belief propagation algorithm and
the Gaussian belief propagation algorithm. Utilizing theoretical results from the Gaussian belief propagation
literature, we extend the LDLC decoder to support
arbitrary column dependent matrices, including non-square matrices.
Meanwhile, algorithms developed for non-parametric belief propagation provides fast and efficient implementation
approximating the message product operation for Gaussian mixtures.
As an example application, we demonstrate that the extended LDLC decoder can be used effectively for solving the
multiuser detection problem, obtaining significant improvements over a previous iterative decoder.}%ignore
\vspace{-1mm}
\section*{Acknowledgment}
D. Bickson would like to thank N. Sommer, M. Feder and Y. Yona from Tel Aviv University for interesting discussions and helpful insights regarding the LDLC algorithm and its implementation. D. Bickson was partially supported by grants NSF IIS-0803333,
NSF NeTS-NBD CNS-0721591
and
DARPA IPTO FA8750-09-1-0141. Danny Dolev is Incumbent of the Berthold Badler Chair in Computer Science. Danny Dolev was  supported in part by the Israeli Science Foundation (ISF) Grant number 0397373.

\scriptsize
\bibliographystyle{IEEEtran}   % (uses file "plain.bst")
\bibliography{Allerton09-2}       % expects file "myrefs.bib"

\end{document}